\documentclass{easychair}

\usepackage{doc}

\usepackage{tikz}
\usepackage{cite}
\usepackage{wrapfig}
\usepackage{amssymb, thmtools, amsthm, amsmath}
\usepackage{skak}
\usepackage{subcaption}
\usepackage{booktabs}
\usepackage{pgfplots}
\usepackage{bbm}
\usepgfplotslibrary{groupplots,statistics}
\pgfplotsset{compat=1.18}

\usepackage{algpseudocode}
\usepackage{algorithm}

\usepackage{cleveref}

\newcommand{\amo}{\textsf{AtMostOne}}
\newcommand{\eo}{\textsf{ExactlyOne}}

\definecolor{qcColor0}{HTML}{174EA6}
\definecolor{qcColor1}{HTML}{E37400}
\definecolor{qcColor2}{HTML}{A50E0E}
\definecolor{qcColor3}{HTML}{6F42C1}
\definecolor{qcColor4}{HTML}{0D652D}

\definecolor{qcOne}{HTML}{66C2A5}   % Soft Teal
\definecolor{qcTwo}{HTML}{FC8D62}   % Soft Orange
\definecolor{qcThree}{HTML}{8DA0CB} % Soft Indigo
\definecolor{qcFour}{HTML}{E78AC3}  % Soft Pink
\definecolor{qcFive}{HTML}{A6D854}  % Lime Green
\definecolor{qcSix}{HTML}{FFD92F}   % Bright Yellow
\definecolor{qcSeven}{HTML}{E5C494} % Tan/Sand

\definecolor{qcInk}{HTML}{172033}
\definecolor{qcBlue}{HTML}{66A5C2}   % Soft Teal
\definecolor{qcBlueLight}{HTML}{66C2A5} 
\definecolor{qcRed}{RGB}{145,40,40}
\definecolor{qcRedLight}{RGB}{213,111,111}
\definecolor{qcGrid}{RGB}{208,213,221}

\newtheorem{lemma}{Lemma}

\title{Coloring Queens with Thousands of Encodings} %TODO Please add

\author{
Bernardo Subercaseaux \and
    Benjamin Przybocki \and Marijn J. H. Heule
}

\institute{
  Carnegie Mellon University,
  Pittsburgh, PA, USA\\
  \email{[bsuberca,bprzyboc,mheule]@andrew.cmu.edu}
 }

\authorrunning{B. Subercaseaux, B. Przybocki and M. Heule}

\titlerunning{Coloring Queens with Thousands of Encodings}

\begin{document}

\maketitle

\begin{abstract}
In \emph{The Art of Computer Programming}, Knuth benchmarked 10 encoding techniques for computing the chromatic number of the queen’s graph: the minimum number of colors needed to color the squares of an $n \times n$ chessboard so that no two squares sharing a row, column, or diagonal receive the same color. In this paper, we extend his analysis much further by comparing thousands of encodings for the same problem, which allows us to identify additional factors that are important for solver performance. We obtain 1584 encodings for this problem by varying (a) the constraints that encode which color is assigned to each cell, (b) the constraints that forbid the same color appearing in a row, column, or diagonal line, and (c) the symmetry-breaking constraints. We find that the three most impactful encoding factors are (i) the choice of symmetry-breaking constraints, (ii) enabling so-called clique hints, and (iii) enforcing that each cell is assigned exactly one color through blocked clauses. Furthermore, while Knuth proposed clique hints as an advantage of the order encoding, we show in fact that they can be effectively employed for the one-hot encoding as well. 
\end{abstract}

% \pagestyle{empty}

% Rainbow colors using pastel shades
\definecolor{qcRed}{HTML}{FFADAD}
\definecolor{qcOrange}{HTML}{FFD6A5}
\definecolor{qcYellow}{HTML}{FDFFB6}
\definecolor{qcGreen}{HTML}{CAFFBF}
\definecolor{qcBlue}{HTML}{9BF6FF}
\definecolor{qcIndigo}{HTML}{A0C4FF}
\definecolor{qcViolet}{HTML}{BDB2FF}

% Helper macro to map modulo results (0-6) to our 7 colors
\newcommand{\getcellcolor}[1]{%
    \ifcase#1 qcSeven\or qcOne\or qcTwo\or qcThree\or qcFour\or qcFive\or qcSix\fi
}

\newcommand{\ModularQueensBoard}{%
    \begin{tikzpicture}[x=0.7cm, y=0.7cm] % Adjust the scale of the board here
        \foreach \i in {1,...,7} {     % Rows (1 to 7)
            \foreach \j in {1,...,7} { % Columns (1 to 7)
                
                % 1. Evaluate the formula: c = (j - 2i) mod 7
                % Note: We add 14 before the modulo to ensure the PGF math engine doesn't return negative remainders.
                \pgfmathsetmacro{\c}{int(mod(\j - 2*\i + 14, 7))}
                
                % 2. Map the result to a color
                \edef\cellcolor{\getcellcolor{\c}}
                
                % 3. Draw and color the square
                \fill[\cellcolor] (\j-1, \i-1) rectangle ++(1, 1);
                \draw[black!65, line width=0.35pt] (\j-1, \i-1) rectangle ++(1, 1);
                
                % 4. Place the queen in the center of the square
                \node at (\j-0.5, \i-0.5) {\fontsize{16pt}{16pt}\selectfont\symqueen};
            }
        }
    \end{tikzpicture}%
}

% \begin{figure}[htpb]
%     \centering
%     \ModularQueensBoard
%     \caption{Illustration of a $7$-coloring for the queen's graph of order $7$}
%     \label{fig:queens-7-modular}
% \end{figure}
\begin{wrapfigure}{r}{0.33\textwidth}
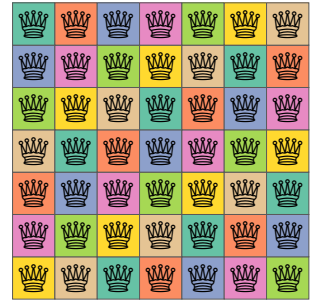

    \centering
    \scalebox{0.8}{
    \ModularQueensBoard
    }
    \caption{A $7$-coloring for the queen's graph of order $7$}
    \label{fig:queen7}
\end{wrapfigure}

\section{Introduction}

There are often many ways to encode a problem as a CNF formula, and how one does this can have a dramatic impact on solver performance. Indeed, improved encodings have been key to many recent applications of SAT to hard mathematical problems~\cite{subercaseauxPackingChromaticNumber2023,happyEnding,unfolding}. Nevertheless, the question of what makes for a good SAT encoding is still poorly understood. In this paper, we shed some light on this problem with a systematic empirical approach: we compare the performance of thousands of encodings for a concrete problem. Specifically, we focus on the problem of determining the chromatic number of the $n \times n$ queen's graph; that is, how many colors do we need to color the squares of an $n \times n$ chessboard so that no two squares sharing a row, column, or diagonal receive the same color (see \Cref{fig:queen7})? More specifically, we ask whether the chromatic number of the $n \times n$ queen's graph is $n$ for $n \in \{8,9,10\}$, which is the smallest it could be \emph{a priori} (the answer in each case is no).
%
% \input{figs/queens_graph7}
% \begin{wrapfigure}{r}{0.33\textwidth}
%     \centering
%     \ModularQueensBoard
%     \caption{A $7$-coloring for the queen's graph of order $7$}
%     \label{fig:queen7}
% \end{wrapfigure}
%

We focus on this problem for a few reasons. First, in Volume 4B of \emph{The Art of Computer Programming}~\cite{taocp-4b}, Knuth used this problem to benchmark various encodings of graph coloring problems, although he only evaluated 10 different encodings, so we thought a more thorough evaluation was warranted. Second, for suitable choices of parameters, the problem lies in a sweet spot of difficulty, taking seconds or minutes to solve with a proper encoding. And third, many constraint satisfaction problems can be naturally cast as graph coloring problems, and the encoding variations we discuss will be applicable to many other such problems.

Most CNF formulas, arising not only from mathematical problems but also from industrial applications, are composed of multiple constraints that may interact with one another. Furthermore, additional constraints can be added on top of base encodings for efficiency (e.g., symmetry breaking). By considering a few options for each of the building blocks of an encoding, one can generate thousands of encodings for a given problem.

For our problem, there are three broad categories of constraints we consider. First, there are the within-cell constraints, which enforce that each cell is assigned at least one color and (optionally) that each cell is assigned at most one color. Second, there are the between-cell constraints, which forbid the same color from appearing in a row, column, or diagonal line. And third, there are the symmetry-breaking constraints, which exploit the symmetries in the problem to make assumptions about how certain cells are colored. Each of the three categories of constraints gives rise to several encoding choices, and we will show that some of these choices have a substantial impact on performance.

Altogether, we consider 1584 encodings for this problem. Our experimental results show that there are a few encoding choices with a large impact on performance. Three factors stand out as especially important. First, a large performance improvement comes from including \emph{clique hints}, which are clauses asserting that every color appears at least once in every clique of size $n$. Clique hints are implied by the graph coloring constraints, but deducing them requires pigeonhole-principle reasoning, which SAT solvers struggle with. The second factor is how one implements symmetry breaking (i.e., for which cells we make assumptions about their color). It is not surprising that symmetry breaking improves performance, but it is interesting how much the symmetry-breaking method affects performance. And third, enforcing that each cell is assigned exactly one color (as opposed to merely at least one color) tends to be helpful. Thus, the practical upshot of our results is that for relevantly similar problems, clique hints should be enabled, one should enforce that each cell is assigned exactly one color, and one should experiment with several different symmetry-breaking constraints to see which one is most performant. In \Cref{sec-results}, we analyze several other encoding choices, some of which have smaller but still measurable impacts on performance and some of which have at most negligible impacts.

We reach a different conclusion than Knuth did in his experiments. Among the encodings that Knuth experimented with, he found an encoding based on the order encoding to perform best, while our results show that encodings based on the one-hot encoding perform best. This is because Knuth only used clique hints in combination with the order encoding following a proposal of Tajima~\cite{tajima}, although they can also be used with the one-hot encoding, yielding a large performance improvement. Another difference is that Knuth did not experiment with different symmetry-breaking constraints, something that our results show to be quite important. In fact, the symmetry breaking Knuth implemented (row-1 symmetry breaking) actually performs worse than most other symmetry-breaking methods we test. Thus, by comparing a wider gamut of encodings, we have refined our understanding of which encoding methods are best and which factors are most important for this problem.

\paragraph{Related work.} In 1848, the chess composer Max Bezzel posed the question of whether it is possible to arrange 8 queens on a standard $8 \times 8$ chessboard such that none can attack another~\cite{Bezzel1848}; the problem was solved affirmatively by Nauck two years later~\cite{Nauck1850}. Graph theoretically, this is equivalent to asking whether the $8 \times 8$ queen's graph has an independent set of size 8. This famous problem has inspired many variations and generalizations~\cite{queens-survey}. The chromatic number of the queen's graph is a natural variant of this famous problem, since each color class forms an independent set. In 1910, Bennett~\cite{Bennett1910} proved that the chromatic number of the $8 \times 8$ queen's graph is strictly greater than 8 (see \cite{Gosset1914} for a simpler proof). With the help of computers, we now know that the chromatic number of the $n \times n$ queen's graph is $n+1$ for $n \in \{8,9,10\}$ and $n$ for $n \in [11,26]$~\cite{vasquez-vimont}. 
Showing that the chromatic number for $n=10$ is $11$ turned out to be hard for many approaches~\cite{petra-1, petra-2, clicolcom, color-column, colorBDD, hebrard2020constraint}, whereas our best encoding can do it in roughly four seconds.  
Related problems regarding the queen's graph (such as its domination number or domatic number) have been attacked via SAT~\cite{rostami2025queendominationsatsolving,domatic}. More generally, SAT is an effective tool for graph coloring problems, and several papers have therefore studied how to do this efficiently~\cite{petra-1,petra-2}.

\section{Encodings}

\subsection{Within-Cell Constraints}

\subsubsection{One-Hot Encoding} \label{sec-one-hot}

For the one-hot encoding, we introduce a variable $c_{(i,j),k}$ for every cell $(i,j)$ and color $k$. The interpretation of $c_{(i,j),k}$ is that the color $k$ is assigned to cell $(i,j)$. According to the problem description, each cell should be assigned \emph{exactly} one color. However, for the encoding to be correct, it is only necessary to enforce that each cell is assigned \emph{at least} one color, since any solution in which some cells are multiply colored can be transformed into one in which every cell is singly colored. If we only enforce that each cell is assigned at least one color, we simply include the clause $\bigvee_k c_{(i,j),k}$ for each cell $(i,j)$.

On the other hand, if we want to enforce that each cell is assigned exactly one color, there are many ways to encode this. The \emph{pairwise} encoding includes both the wide clauses of the form $\bigvee_k c_{(i,j),k}$ as well as the clauses $\overline{c_{(i,j),k}} \lor \overline{c_{(i,j),k'}}$ for all cells $(i,j)$ and distinct colors $k$ and $k'$. Note that the binary clauses are \emph{blocked} with respect to the previous clauses of the form $\bigvee_k c_{(i,j),k}$~\cite{Kullman99}.
Alternatively, we can introduce auxiliary variables to encode the same constraint using fewer clauses. The encodings we consider are all slight variations of the sequential counter encoding~\cite{sinz}. Given $\ell \ge 3$, we can recursively encode the at-most-one constraint as follows:
\begin{equation} \label{eq-rec}
    \amo(x_1,\dots,x_n) := \amo(x_1,\dots,x_{\ell-1},\overline{y}) \land \amo(y,x_{\ell},\dots,x_n),
\end{equation}
where we use the pairwise encoding as a base case when $n \le \ell$. We call $\ell$ the \emph{cut} parameter, and we will only consider $\ell \in \{3,4,5\}$. Together with the wide clauses of the form $\bigvee_k c_{(i,j),k}$, this yields an encoding for the exactly one constraint that we can apply to $\{c_{(i,j),k} \mid k \in [n]\}$. We can also omit the wide clauses and instead recursively encode the exactly one constraint with the same recursion as in \eqref{eq-rec}, with the only difference being that in the base case ($n \le \ell$), we encode $\eo(x_1,\dots,x_n)$ using the pairwise encoding together with the clause $x_1 \lor \dots \lor x_n$. This way, all clauses have width at most $\ell$, and thus we call this the \emph{narrow} variant. We can obtain yet another variation by substituting $\amo(x_{\ell},\dots,x_n,y)$ for $\amo(y,x_{\ell},\dots,x_n)$ in \eqref{eq-rec} (and similarly for $\eo$); following the terminology of~\cite{nawrockiXORLocalSearch2021}, we call the original variant with $y$ as the first argument the \emph{linear} encoding, and we call the variant with $y$ as the last argument the \emph{pooled} encoding. Finally, when we apply one of the above at-most-one or exactly one encodings to $\{c_{(i,j),k} \mid k \in [n]\}$, we can either apply it to $c_{(i,j),1}, \dots, c_{(i,j),n}$ in order or we can randomly shuffle the order of these variables beforehand (independently for each cell).

In summary, we have three strategies for encoding the within-cell constraints with the one-hot encoding: the at-least-one constraint, the exactly one constraint using the pairwise encoding, and the exactly one constraint using one of the recursive encodings. For the recursive encodings, we vary four parameters: the choice of $\ell \in \{3,4,5\}$, the choice of using wide or narrow clauses for at-least-one, the choice of linear versus pooled, and the choice of sorted versus randomized variable ordering. In total, this yields 26 one-hot encodings for the within-cell constraints.

\subsubsection{Order Encoding}

Instead of introducing the variables $c_{(i,j),k}$, the \emph{order} encoding~\cite{crawford-baker,order,ulc} uses variables of the form $o_{(i,j),k}$ for each cell $(i,j)$ and $k \in [n-1]$. The standard interpretation of $o_{(i,j),k}$ is that the cell $(i,j)$ is assigned one of the colors from $\{1,\dots,k\}$. To enforce this semantics, we include the clauses $\overline{o_{(i,j),k-1}} \lor o_{(i,j),k}$. Then, we can assert that cell $(i,j)$ has color $k \in [2,n-1]$ with the conjunction $\overline{o_{(i,j),k-1}} \land o_{(i,j),k}$; we use $o_{(i,j),1}$ to assert that $(i,j)$ has color 1 and $\overline{o_{(i,j),n-1}}$ to assert that $(i,j)$ has color $n$.
Note that the order encoding variables implicitly guarantee that at least one color is chosen: each cell $(i,j)$ is assigned the color $k$, where $k$ is the smallest integer such that $o_{(i,j),k}$ is true, or $k=n$ if there is no such $k$.
In the rest of this paper, we will sometimes use the literals $c_{(i,j),k}$ and $\overline{c_{(i,j),k}}$ with the understanding that if we are using the order encoding, these must be replaced by their definitions in terms of the order variables.

When using the order encoding for the within-cell constraints, there are two parameters we vary, yielding four order encodings for the within-cell constraints. First, we can drop the clauses $\overline{o_{(i,j),k-1}} \lor o_{(i,j),k}$, which are in fact unnecessary. This allows there to be multiple colors $k$ for which $\overline{o_{(i,j),k-1}} \land o_{(i,j),k}$. This is analogous to not including the at-most-one constraint on colors for the one-hot encoding. Second, instead of using the standard order on colors $1,\dots,n$, we can instead implement the order encoding on a random ordering of the colors (independently for each cell). That is, if $\sigma \colon [n] \to [n]$ is a random permutation, then the interpretation of $o_{(i,j),k}$ is now that $(i,j)$ is assigned one of the colors from $\{\sigma(1), \dots, \sigma(k)\}$. This interpretation can be encoded with only minor modifications to the above description.

\subsection{Between-Cell Constraints}

\subsubsection{Independence Constraint}

The independence constraint asserts that cells lying in the same row, column, or diagonal line have distinct colors. The simplest way to encode this is the \emph{pairwise} encoding, which includes a clause of the form $\overline{c_{(i,j),k}} \lor \overline{c_{(i',j'),k}}$ for each color $k$ and each pair of distinct cells $\{(i,j),(i',j')\}$ lying in the same row, column, or diagonal line. But we can also use one of the recursive at-most-one encodings described in \Cref{sec-one-hot} to enforce that each color occurs at most once along each row, column, and diagonal line. To avoid a combinatorial explosion in the number of encodings we consider, if we use a recursive at-most-one encoding for the within-cell constraints and the independence constraints, we use the same cut $\ell \in \{3,4,5\}$, the same choice of linear versus pooled, and the same choice of sorted versus randomized variable ordering.

\subsubsection{Clique Hints}

For each clique of size $n$ in the queen's graph (i.e., each row, column, the main diagonal, and the antidiagonal), we know that each color must appear exactly once in this clique. Thus, we can optionally add \emph{clique hints} to our encoding, which assert that every color appears at least once in every clique of size $n$. If we use one of the recursive at-most-one encodings for the independence constraints, we have two ways to encode a clique hint: we can use either wide or narrow clauses for the at-least-one constraint, as described in \Cref{sec-one-hot}. For simplicity, when both clique hints and within-cell constraints are encoded recursively, we take the same setting (wide or narrow) for both.

Knuth attributes the idea of applying clique hints to this problem to Tajima~\cite{tajima}, although their implementation of clique hints differs from ours and only applies to the order encoding. Specifically, for every row, column, and diagonal line of size $k$ (not merely the main diagonal and antidiagonal), they assert that some cell has a color $\le n-k+1$ and some cell has a color $\ge k$.\footnote{Knuth also considers a \emph{double clique hint}, which imposes these order-based clique hints for two different permutations of the colors.} In particular, for each clique of size $n$, their constraint asserts that the colors $1$ and $n$ both appear, which is much weaker than our version of the clique hint. Their version has the advantage that it can be applied to all diagonal lines, and when using the order encoding, it doesn't require introducing auxiliary variables to play the role of $c_{(i,j),k}$. However, we found that our implementation of the clique hints considerably outperforms theirs, so we did not consider their version further.

We remark that clique hints can be generalized to cases in which the number of colors exceeds the clique size:

\begin{lemma}\label{lemma:clique-hints}
     Let $G$ be any graph, and $K \subseteq V(G)$ be a clique in $G$. Then, for any proper coloring $f\colon V(G) \to \{1, \dots, t\}$,  and every set $X \subseteq \{1, \dots, t\}$, we have
    \[
       \sum_{v \in K} \sum_{c \in X} \mathbbm{1}_{\{f(v) = c \}} \geq |K| + |X| - t.
    \]
\end{lemma}
\begin{proof}
Let us say a vertex $v \in K$ is ``good'' if $f(v) \in X$. Assume, for the sake of a contradiction, that the number of good vertices is at most $|K| + |X| - t - 1$. Then, there are at least $|K| - (|K| + |X| - t - 1) = t - |X| + 1$ non-good vertices in $K$. Since these non-good vertices induce a clique, and $f$ assigns a color in $\{1, \dots, t\} \setminus X$ to each, we have a $(t- |X|)$-coloring of a $(t-|X|+1)$-clique, which is a contradiction.
\end{proof}
   
The clique hints described above correspond to the case $|K| = t$ and $|X| = 1$ of~\Cref{lemma:clique-hints}. We experimented as well with $|K| = t-1$, by considering for $n = 9$ the four off-by-one diagonals of size $8$, and adding for each pair of colors $(i, j)$ with $1 \leq i < j \leq n$ a wide clause stating that either some vertex of that diagonal gets color $i$ or some vertex gets color $j$, but it turned out to be detrimental in practice. We thus limit our analysis of clique hints to the case $|X| = 1$.

\subsection{Symmetry Breaking}

Symmetry-breaking constraints reduce the search space by making assumptions about how certain cells are colored. There are two kinds of symmetries to break: the $n!$ symmetries obtained by permuting the colors ($S_n$) and the 8 symmetries obtained by rotating and reflecting the board ($D_4$). Breaking all $8 \cdot n!$ symmetries effectively may be difficult, so we will be content with trying to break most of them.

A simple symmetry-breaking constraint is to assume that the first row is colored $1, 2, \dots, n$, which is valid since the colors in any solution could be relabeled to satisfy this constraint. This breaks the color symmetries but not the board symmetries. We can also impose the same constraint for any of the first $\lceil n/2 \rceil$ rows or for the main diagonal. But there are also more subtle ways to break symmetries. We consider two additional strategies for odd $n$ (e.g., $n=9$), where we first fix 5 colors in the center of the board in the shape of either a plus or a cross (see \Cref{fig:plus-cross}). The idea of using a plus shape for symmetry breaking is from Heule, Karahalios, and van Hoeve~\cite{clicolcom}, who observed that this is very effective. Next, we describe how to impose some additional constraints on the remaining 4 cells in the $3 \times 3$ grid containing the plus or cross. The goal is to add clauses so that each valid coloring of the $3 \times 3$ grid is represented by exactly one satisfying assignment up to the symmetries (both the $S_n$ and $D_4$ symmetries). To do this, we add clauses that reduce the number of satisfying assignments while preserving the invariant that every valid coloring is represented by some satisfying assignment. We consider clauses of increasing width (starting with unit clauses), adding them to the encoding if and only if they meet the previous condition, and we stop once every valid coloring is represented by exactly one satisfying assignment. More precisely, given a CNF formula $\varphi$ and $X \subseteq \mathrm{Vars}(\varphi)$, we assume that we have a function $\mathsf{enum}(\varphi, X)$ that returns the set of partial models of $\varphi$ over the restricted domain $X$; that is, $\mathsf{enum}(\varphi, X)$ consists of the assignments to $X$ that satisfy the clauses of $\varphi$ only involving variables in $X$. Given a set of partial models $M$, we also assume a function $\mathsf{canon}(M)$, which puts every model into a `canonical' form, so that two models are isomorphic if and only if their canonical forms are equal. Then, the concrete algorithm we use is presented in~\Cref{alg:sym-break}. We instantiate it with $X$ being all the variables $c_{(i, j), k}$ for which $(i, j)$ belongs to the central $3 \times 3$ grid.

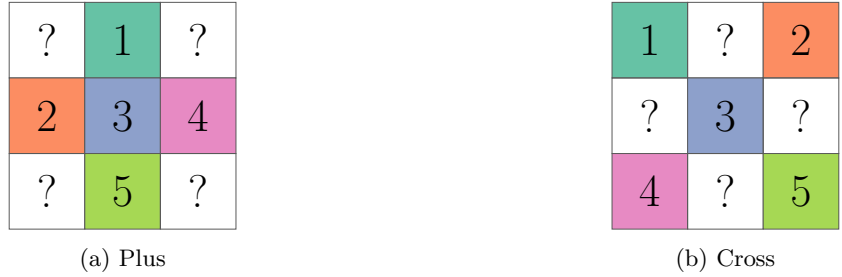
\begin{figure}[ht]
    \centering
    \begin{subfigure}{0.45\textwidth}
        \centering
        \begin{tikzpicture}[x=1cm, y=1cm] % Adjust the overall scale here
        
        % ==========================================
        % ROW 1 (Top)
        % ==========================================
        
        % Cell 1
        \draw[black!65, line width=0.35pt] (0, 2) rectangle ++(1, 1);
        \node at (0.5, 2.5) {\fontsize{16pt}{16pt}\selectfont ?};
        
        % Cell 2
        \fill[qcOne] (1, 2) rectangle ++(1, 1);
        \draw[black!65, line width=0.35pt] (1, 2) rectangle ++(1, 1);
        \node at (1.5, 2.5) {\fontsize{16pt}{16pt}\selectfont 1};
        
        % Cell 3
        \draw[black!65, line width=0.35pt] (2, 2) rectangle ++(1, 1);
        \node at (2.5, 2.5) {\fontsize{16pt}{16pt}\selectfont ?};

        % ==========================================
        % ROW 2 (Middle)
        % ==========================================
        
        % Cell 4
        \fill[qcTwo] (0, 1) rectangle ++(1, 1);
        \draw[black!65, line width=0.35pt] (0, 1) rectangle ++(1, 1);
        \node at (0.5, 1.5) {\fontsize{16pt}{16pt}\selectfont 2};
        
        % Cell 5
        \fill[qcThree] (1, 1) rectangle ++(1, 1);
        \draw[black!65, line width=0.35pt] (1, 1) rectangle ++(1, 1);
        \node at (1.5, 1.5) {\fontsize{16pt}{16pt}\selectfont 3};
        
        % Cell 6
        \fill[qcFour] (2, 1) rectangle ++(1, 1);
        \draw[black!65, line width=0.35pt] (2, 1) rectangle ++(1, 1);
        \node at (2.5, 1.5) {\fontsize{16pt}{16pt}\selectfont 4};

        % ==========================================
        % ROW 3 (Bottom)
        % ==========================================
        
        % Cell 7
        \draw[black!65, line width=0.35pt] (0, 0) rectangle ++(1, 1);
        \node at (0.5, 0.5) {\fontsize{16pt}{16pt}\selectfont ?};
        
        % Cell 8
        \fill[qcFive] (1, 0) rectangle ++(1, 1);
        \draw[black!65, line width=0.35pt] (1, 0) rectangle ++(1, 1);
        \node at (1.5, 0.5) {\fontsize{16pt}{16pt}\selectfont 5};
        
        % Cell 9
        \draw[black!65, line width=0.35pt] (2, 0) rectangle ++(1, 1);
        \node at (2.5, 0.5) {\fontsize{16pt}{16pt}\selectfont ?};

        \end{tikzpicture}
        \caption{Plus}
        \label{fig:plus}
    \end{subfigure}
    \hfill
    \begin{subfigure}{0.45\textwidth}
        \centering
        \begin{tikzpicture}[x=1cm, y=1cm] % Adjust the overall scale here
        
        % ==========================================
        % ROW 1 (Top)
        % ==========================================
        
        % Cell 1
        \fill[qcOne] (0, 2) rectangle ++(1, 1);
        \draw[black!65, line width=0.35pt] (0, 2) rectangle ++(1, 1);
        \node at (0.5, 2.5) {\fontsize{16pt}{16pt}\selectfont 1};
        
        % Cell 2
        \draw[black!65, line width=0.35pt] (1, 2) rectangle ++(1, 1);
        \node at (1.5, 2.5) {\fontsize{16pt}{16pt}\selectfont ?};
        
        % Cell 3
        \fill[qcTwo] (2, 2) rectangle ++(1, 1);
        \draw[black!65, line width=0.35pt] (2, 2) rectangle ++(1, 1);
        \node at (2.5, 2.5) {\fontsize{16pt}{16pt}\selectfont 2};

        % ==========================================
        % ROW 2 (Middle)
        % ==========================================
        
        % Cell 4
        \draw[black!65, line width=0.35pt] (0, 1) rectangle ++(1, 1);
        \node at (0.5, 1.5) {\fontsize{16pt}{16pt}\selectfont ?};
        
        % Cell 5
        \fill[qcThree] (1, 1) rectangle ++(1, 1);
        \draw[black!65, line width=0.35pt] (1, 1) rectangle ++(1, 1);
        \node at (1.5, 1.5) {\fontsize{16pt}{16pt}\selectfont 3};
        
        % Cell 6
        \draw[black!65, line width=0.35pt] (2, 1) rectangle ++(1, 1);
        \node at (2.5, 1.5) {\fontsize{16pt}{16pt}\selectfont ?};

        % ==========================================
        % ROW 3 (Bottom)
        % ==========================================
        
        % Cell 7
        \fill[qcFour] (0, 0) rectangle ++(1, 1);
        \draw[black!65, line width=0.35pt] (0, 0) rectangle ++(1, 1);
        \node at (0.5, 0.5) {\fontsize{16pt}{16pt}\selectfont 4};
        
        % Cell 8
        \draw[black!65, line width=0.35pt] (1, 0) rectangle ++(1, 1);
        \node at (1.5, 0.5) {\fontsize{16pt}{16pt}\selectfont ?};
        
        % Cell 9
        \fill[qcFive] (2, 0) rectangle ++(1, 1);
        \draw[black!65, line width=0.35pt] (2, 0) rectangle ++(1, 1);
        \node at (2.5, 0.5) {\fontsize{16pt}{16pt}\selectfont 5};

        \end{tikzpicture}
        \caption{Cross}
        \label{fig:cross}
    \end{subfigure}
    \caption{Plus and cross symmetry breaking}
    \label{fig:plus-cross}
\end{figure}

\begin{algorithm}
  \caption{Symmetry Breaking Generation}
  \begin{algorithmic}[1]
    \Require $\varphi$ is a CNF formula and $X \subseteq \mathrm{Vars}(\varphi)$ is a restricted domain.
    \Procedure{Symmetry Breaking Generation}{$\varphi, X$}
    \State $M \gets \mathsf{enum}(\varphi, X)$
    \State $M_{\textsf{can}} \gets \mathsf{canon}(M)$
    \State $\ell \gets 1$
    \While{$|M_{\textsf{can}}| < |M|$}
    \For{each clause $C$ of length $\ell$ with $\mathrm{Vars}(C) \subseteq X$}
        \State $M' \gets\mathsf{enum}(\varphi \land C, X)$
        \State $M'_{\textsf{can}}  \gets\mathsf{canon}(M')$
        \If{$M_{\textsf{can}} = M'_{\textsf{can}}$ and $|M'| < |M|$}
            \State $\varphi \gets \varphi \land C$
            \State $M \gets M'$
            \State $M_{\textsf{can}}\gets M'_{\textsf{can}}$
        \EndIf
    \EndFor
    \State $\ell \gets \ell + 1$
    \EndWhile
    \State \Return $\varphi$
    \EndProcedure
  \end{algorithmic}\label{alg:sym-break}
\end{algorithm}

For instance, the concrete clauses we obtain for $n=9$ in the plus configuration are
\begin{align*}
    &\bigwedge_{k \in \{4,6,7,8\}} \overline{c_{(4,4), k}} \land \bigwedge_{k \in \{5,6,7,9\}} \overline{c_{(4,6), k}} \land \bigwedge_{k \in \{1,6,8,9\}} \overline{c_{(6,4), k}} \land \bigwedge_{k \in \{7,8,9\}} \overline{c_{(6,6), k}}\\  &{\land} \, (c_{(4, 4), 5} \lor \overline{c_{(4, 6), 2}}) \land (c_{(4, 4), 5} \lor \overline{c_{(6, 6), 1}}) \land (\overline{c_{(4, 6), 2}} \lor c_{(6, 4), 4})  \land (\overline{c_{(4, 6), 2}} \lor c_{(6, 6), 1}) \\
    &{\land} \, (c_{(4,4), 5} \lor c_{(6, 4), 4} \lor \overline{c_{(6,6), 2}}) \land (c_{(4,4), 5} \lor \overline{c_{(6, 4), 4}} \lor c_{(6,6), 2}).
\end{align*}
  
We verified the correctness of all the symmetry-breaking clauses described above using the SR proof format~\cite{CodelAH24}.

% ['x_(0, 0)_4', 'x_(2, 0)_3', '-x_(2, 2)_1']
% ['x_(0, 0)_4', '-x_(2, 0)_3', 'x_(2, 2)_1']

\section{Results} \label{sec-results}

We focus our main analysis on the $n=9$ and $n=10$ cases, since the $n=8$ case seems to be too easy to thoroughly evaluate differences in encodings (with a proper encoding, the problem can be solved in under 20 milliseconds). We always enable clique hints and some form of symmetry breaking, since this is necessary to run all encoding options to completion with a small timeout (10 minutes). This results in considering 1584 encodings for $n=9$ and $1188$ for $n=10$; the difference is that plus and cross symmetry breaking are only applicable for odd $n$.

To show the effect of disabling clique hints or symmetry breaking on $n=8$, we ran the worst combination of encoding parameters (based on the data using clique hints and symmetry breaking) without clique hints or symmetry breaking, and it took $76.82$ core hours, computed with a cube-and-conquer split (16299 cubes) using \texttt{march\_cu}~\cite{HeuleKullmannWieringaBiere11} on the Bridges-2 cluster~\cite{cluster}. In general, as $n$ increases, the impact of symmetry breaking trumps that of clique hints (which can be explained by the $n!$ theoretical speed-up based on symmetries); for $n = 9$, instances using symmetry breaking (plus) but no clique hints can be solved in a matter of minutes (0.21 core hours in the worst case), whereas instances with clique hints but no symmetry breaking took multiple hours (11.02 core hours in the worst case).

We ran all of our experiments on 8 threads on a 2024 MacBook Pro with an M4 Pro CPU and 24 GB of RAM. We used the solver \textsf{Kissat}~\cite{BiereFallerFazekasFleuryFroleyksPollitt-SAT-Competition-2024-solvers} (v4.0.4 \texttt{8af8e56}), running each instance on a single thread (i.e., no cube-and-conquer) with the \texttt{--no-factor} flag to avoid Bounded Variable Addition reencodings.  Our code and data are publicly available at \url{https://github.com/bsubercaseaux/QueenColoringEncodings}. Furthermore, we have designed an interactive website where the runtimes for different encoding settings can be compared visually; we encourage the reader to visit it at \url{https://queen-encodings.surge.sh}. 

\subsection{Regression Models}\label{subsec:regression}

For $n=9$, our runtimes ranged from 0.048 to 3.664 seconds, and for $n=10$, our runtimes ranged from 4.467 to 538.101 seconds. To get an understanding of this variance and how each factor contributes to the runtime, we ran a linear regression on the logarithm of the runtimes. The results are presented in \Cref{tab:regression-9,tab:regression-10}. The options whose coefficients are listed as an em dash (---) are treated as the default options for the sake of the model, and each remaining option is modeled using a categorical predictor variable. The model predicts the natural logarithm of the runtime as a linear function of these categorical predictors. Since our goal is to compare the performance of encoding options, we omit the constant coefficient in the tables; in other words, we only predict the difference in runtime relative to the default options. For example, the coefficient for row-2 symmetry breaking is $-0.2680$, which means that the predicted effect of switching from row-1 symmetry breaking (the baseline) to row 2 is to multiply the runtime by $e^{-0.2680} \approx 0.76491$.
Note that the $R^2$ for both models is very high, so a linear model explains most of the variance.

From the regression models, we can see which factors are most important: the choice of symmetry breaking and the decision to use one of the exactly-one encodings for the within-cell constraints. On the other hand, the details of how the recursive encodings are implemented do not have a large impact (for both within- and between-cell constraints). Randomizing the variable ordering yields a small but significant improvement for the within-cell constraints, while having a negligible effect for the between-cell constraints. We also see that the order encoding is not as performant as the one-hot encoding when we use the best settings for each. Finally, it is interesting to note the differences between \Cref{tab:regression-9,tab:regression-10}. For example, diagonal symmetry breaking is good relative to row-1 symmetry breaking for $n=9$ but not for $n=10$. This shows that what makes for a good encoding is highly sensitive to the choice of problem, so one should be careful extrapolating our results to other problems.

\begin{table}[p]
    \centering
    \caption{Linear regression  for the natural logarithm of the runtimes for $n = 9$ ($R^2 = 0.937$)}
\begin{tabular}{lcc}
\toprule
 & coefficient & confidence interval (95\%) \\
\midrule
\textbf{Symmetry breaking} & & \\
Row symmetry breaking & & \\
\quad Row 1 \textbf{(baseline)} & --- & --- \\
\quad Row 2 & $-0.2680$ & $[-0.307, -0.229]$ \\
\quad Row 3 & $-0.2515$ & $[-0.290, -0.213]$ \\
\quad Row 4 & $-0.3009$ & $[-0.340, -0.262]$ \\
\quad Row 5 & $-0.5474$ & $[-0.586, -0.509]$ \\
Diagonal symmetry breaking & $-0.3969$ & $[-0.436, -0.358]$ \\
Plus symmetry breaking & $-1.5139$ & $[-1.553, -1.475]$ \\
Cross symmetry breaking & $-1.0350$ & $[-1.074, -0.996]$ \\
\midrule
\textbf{Within-cell constraints} & & \\
One-hot & & \\
\quad At least one \textbf{(baseline)} & --- & --- \\
\quad Exactly one (pairwise) & $-0.7280$ & $[-0.767, -0.689]$ \\
\quad Exactly one (recursive) & $-0.5414$ & $[-0.598, -0.485]$ \\
\qquad Cut & & \\
\qquad \quad Cut 3 & --- & --- \\
\qquad \quad Cut 4 & $-0.0624$ & $[-0.112, -0.012]$ \\
\qquad \quad Cut 5 & $-0.0524$ & $[-0.102, -0.002]$ \\
\qquad Auxiliary placement & & \\
\qquad \quad Linear & --- & --- \\
\qquad \quad Pooled & $+0.0021$ & $[-0.039, +0.043]$ \\
\qquad At least one type & & \\
\qquad \quad Wide & --- & --- \\
\qquad \quad Narrow & $-0.0436$ & $[-0.084, -0.003]$ \\
Order & $+0.8547$ & $[+0.819, +0.890]$ \\
\quad At least one & --- & --- \\
\quad Exactly one & $-0.9124$ & $[-0.940, -0.885]$ \\
Literal ordering & & \\
\quad Sorted & --- & --- \\
\quad Randomized & $-0.1340$ & $[-0.157, -0.111]$ \\
\midrule
\textbf{Between-cell constraints} & & \\
Pairwise \textbf{(baseline)} & --- & --- \\
Recursive & $+0.1313$ & $[+0.091, +0.172]$ \\
\quad Cut & & \\
\qquad Cut 3 & --- & --- \\
\qquad Cut 4 & $+0.0048$ & $[-0.022, +0.032]$ \\
\qquad Cut 5 & $+0.0008$ & $[-0.026, +0.028]$ \\
\quad Auxiliary placement & & \\
\qquad Linear & --- & --- \\
\qquad Pooled & $-0.0255$ & $[-0.047, -0.004]$ \\
\quad Literal ordering & & \\
\qquad Sorted & --- & --- \\
\qquad Randomized & $+0.0231$ & $[+0.002, +0.044]$ \\
\quad Clique hint type & & \\
\qquad Wide & --- & --- \\
\qquad Narrow & $+0.0066$ & $[-0.015, +0.028]$ \\
\bottomrule
\end{tabular}
    \label{tab:regression-9}
\end{table}

\begin{table}[p]
    \centering
    \caption{Linear regression for the natural logarithm of the runtimes for $n = 10$ ($R^2 = 0.941$)}
\begin{tabular}{lcc}
\toprule
 & coefficient & confidence interval (95\%) \\
\midrule
\textbf{Symmetry breaking} & & \\
Row symmetry breaking & & \\
\quad Row 1 \textbf{(baseline)} & --- & --- \\
\quad Row 2 & $-0.0762$ & $[-0.127, -0.025]$ \\
\quad Row 3 & $-0.1042$ & $[-0.155, -0.053]$ \\
\quad Row 4 & $-0.2002$ & $[-0.251, -0.149]$ \\
\quad Row 5 & $-0.2897$ & $[-0.341, -0.239]$ \\
Diagonal symmetry breaking & $+0.0660$ & $[+0.015, +0.117]$ \\
\midrule
\textbf{Within-cell constraints} & & \\
One-hot & & \\
\quad At least one \textbf{(baseline)} & --- & --- \\
\quad Exactly one (pairwise) & $-2.3251$ & $[-2.384, -2.266]$ \\
\quad Exactly one (recursive) & $-2.0789$ & $[-2.164, -1.994]$ \\
\qquad Cut & & \\
\qquad \quad Cut 3 & --- & --- \\
\qquad \quad Cut 4 & $+0.0099$ & $[-0.066, +0.086]$ \\
\qquad \quad Cut 5 & $-0.0092$ & $[-0.085, +0.067]$ \\
\qquad Auxiliary placement & & \\
\qquad \quad Linear & --- & --- \\
\qquad \quad Pooled & $+0.0481$ & $[-0.014, +0.110]$ \\
\qquad At least one type & & \\
\qquad \quad Wide & --- & --- \\
\qquad \quad Narrow & $+0.1524$ & $[+0.090, +0.215]$ \\
Order & $+0.1412$ & $[+0.088, +0.195]$ \\
\quad At least one & --- & --- \\
\quad Exactly one & $-1.8553$ & $[-1.897, -1.814]$ \\
Literal ordering & & \\
\quad Sorted & --- & --- \\
\quad Randomized & $-0.2027$ & $[-0.237, -0.168]$ \\
\midrule
\textbf{Between-cell constraints} & & \\
Pairwise \textbf{(baseline)} & --- & --- \\
Recursive & $+0.0543$ & $[-0.007, +0.116]$ \\
\quad Cut & & \\
\qquad Cut 3 & --- & --- \\
\qquad Cut 4 & $-0.0042$ & $[-0.045, +0.036]$ \\
\qquad Cut 5 & $-0.0167$ & $[-0.057, +0.024]$ \\
\quad Auxiliary placement & & \\
\qquad Linear & --- & --- \\
\qquad Pooled & $+0.0236$ & $[-0.010, +0.057]$ \\
\quad Literal ordering & & \\
\qquad Sorted & --- & --- \\
\qquad Randomized & $+0.0099$ & $[-0.023, +0.042]$ \\
\quad Clique hint type & & \\
\qquad Wide & --- & --- \\
\qquad Narrow & $-0.0308$ & $[-0.064, +0.002]$ \\
\bottomrule
\end{tabular}
    \label{tab:regression-10}
\end{table}

\subsection{Symmetry Breaking}

Given the importance of differences in symmetry-breaking constraints, we also show in \Cref{fig:sb-plots} box-and-whisker plots of the runtimes with different choices of symmetry-breaking constraints for both $n=9$ and $n=10$. For $n=9$, we again see that the plus constraint is the best, with the cross being a close second.  
These plots also provide some nuance to the summary in \Cref{tab:regression-9,tab:regression-10}. For example, for $n=10$, \Cref{tab:regression-10} suggests that symmetry breaking on row 5 is noticeably better than on row 1, and yet \Cref{fig:sb-plots} shows that the median runtime for row-1 symmetry breaking is actually lower than row 5. Thus, for $n=10$, row-5 symmetry breaking is better on average than row 1 because the latter has a longer tail of large runtimes than the former.

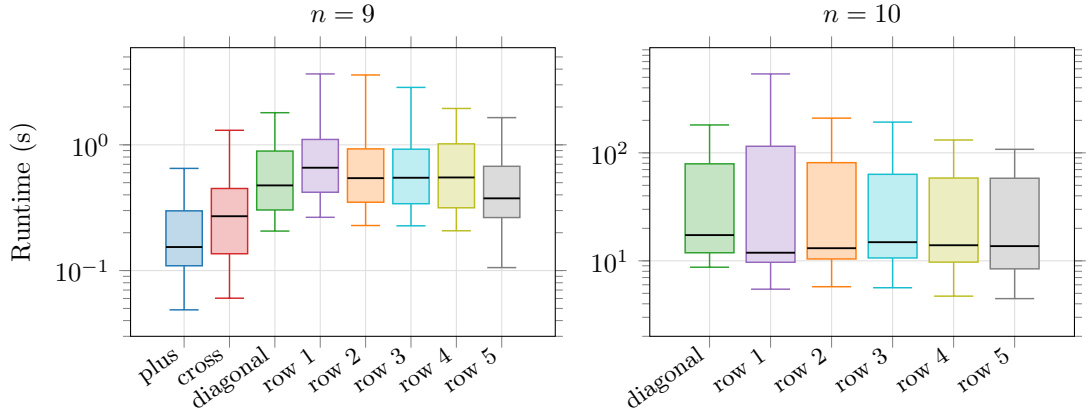
\begin{figure}[ht]
    \centering
   \definecolor{symPlus}{RGB}{31,119,180}
\definecolor{symCross}{RGB}{214,39,40}
\definecolor{symDiagonal}{RGB}{44,160,44}
\definecolor{symRowZero}{RGB}{148,103,189}
\definecolor{symRowOne}{RGB}{255,127,14}
\definecolor{symRowTwo}{RGB}{23,190,207}
\definecolor{symRowThree}{RGB}{188,189,34}
\definecolor{symRowFour}{RGB}{127,127,127}

\begin{tikzpicture}
\begin{groupplot}[
  group style={group size=2 by 1, horizontal sep=1.25cm},
  width=7.2cm,
  height=5.4cm,
  ymin=0.03,
  boxplot/draw direction=y,
  xtick={1,2,3,4,5,6,7,8},
  xticklabels={plus,cross,diagonal,row 1,row 2,row 3,row 4,row 5},
  x tick label style={rotate=35, anchor=east, font=\small},
  ylabel={Runtime (s)},
  grid=major,
  ymode=log,
  grid style={draw=gray!20},
  major grid style={draw=gray!25},
  tick align=outside,
  boxplot/every box/.style={solid, line width=0.55pt},
  boxplot/every whisker/.style={solid, line width=0.55pt},
  boxplot/every median/.style={solid, draw=black, line width=0.7pt},
]
\nextgroupplot[title={$n=9$}, ylabel={Runtime (s)}]
\addplot [draw=symPlus, fill=symPlus!28, solid, line width=0.55pt,
  boxplot prepared={
    draw position=1,
    lower whisker=0.048716375,
    lower quartile=0.1091858,
    median=0.15391223,
    upper quartile=0.29849773,
    upper whisker=0.64949783
  }
] coordinates {};
\addplot [draw=symCross, fill=symCross!28, solid, line width=0.55pt,
  boxplot prepared={
    draw position=2,
    lower whisker=0.06029675,
    lower quartile=0.13608448,
    median=0.27051102,
    upper quartile=0.44955146,
    upper whisker=1.3052195
  }
] coordinates {};
\addplot [draw=symDiagonal, fill=symDiagonal!28, solid, line width=0.55pt,
  boxplot prepared={
    draw position=3,
    lower whisker=0.20612479,
    lower quartile=0.3033123,
    median=0.47565096,
    upper quartile=0.89215562,
    upper whisker=1.80101
  }
] coordinates {};
\addplot [draw=symRowZero, fill=symRowZero!28, solid, line width=0.55pt,
  boxplot prepared={
    draw position=4,
    lower whisker=0.26522046,
    lower quartile=0.41969361,
    median=0.65763031,
    upper quartile=1.1036377,
    upper whisker=3.6642759
  }
] coordinates {};
\addplot [draw=symRowOne, fill=symRowOne!28, solid, line width=0.55pt,
  boxplot prepared={
    draw position=5,
    lower whisker=0.228178,
    lower quartile=0.34966996,
    median=0.54326581,
    upper quartile=0.92873857,
    upper whisker=3.5944511
  }
] coordinates {};
\addplot [draw=symRowTwo, fill=symRowTwo!28, solid, line width=0.55pt,
  boxplot prepared={
    draw position=6,
    lower whisker=0.22697683,
    lower quartile=0.34014052,
    median=0.5473265,
    upper quartile=0.92304729,
    upper whisker=2.8628732
  }
] coordinates {};
\addplot [draw=symRowThree, fill=symRowThree!28, solid, line width=0.55pt,
  boxplot prepared={
    draw position=7,
    lower whisker=0.2072425,
    lower quartile=0.31568226,
    median=0.55013654,
    upper quartile=1.0182554,
    upper whisker=1.9460709
  }
] coordinates {};
\addplot [draw=symRowFour, fill=symRowFour!28, solid, line width=0.55pt,
  boxplot prepared={
    draw position=8,
    lower whisker=0.10558929,
    lower quartile=0.26398758,
    median=0.37535977,
    upper quartile=0.67496302,
    upper whisker=1.6456619
  }
] coordinates {};

\nextgroupplot[title={$n=10$}, ylabel={}, ymin=2]
% \addplot [draw=symPlus, fill=symPlus!28, solid, line width=0.55pt,
%   boxplot prepared={
%     draw position=1,
%     lower whisker=2.2747706,
%     lower quartile=3.6128805,
%     median=5.3112873,
%     upper quartile=38.620109,
%     upper whisker=122.45294
%   }
% ] coordinates {};
% \addplot [draw=symCross, fill=symCross!28, solid, line width=0.55pt,
%   boxplot prepared={
%     draw position=2,
%     lower whisker=3.1741985,
%     lower quartile=5.0053407,
%     median=7.3834875,
%     upper quartile=45.816594,
%     upper whisker=133.43069
%   }
% ] coordinates {};
\addplot [draw=symDiagonal, fill=symDiagonal!28, solid, line width=0.55pt,
  boxplot prepared={
    draw position=3,
    lower whisker=8.7291631,
    lower quartile=11.876424,
    median=17.313075,
    upper quartile=79.135513,
    upper whisker=181.22946
  }
] coordinates {};
\addplot [draw=symRowZero, fill=symRowZero!28, solid, line width=0.55pt,
  boxplot prepared={
    draw position=4,
    lower whisker=5.4634137,
    lower quartile=9.6951756,
    median=11.888714,
    upper quartile=115.05386,
    upper whisker=538.10142
  }
] coordinates {};
\addplot [draw=symRowOne, fill=symRowOne!28, solid, line width=0.55pt,
  boxplot prepared={
    draw position=5,
    lower whisker=5.7592821,
    lower quartile=10.410505,
    median=13.089793,
    upper quartile=81.091447,
    upper whisker=209.56147
  }
] coordinates {};
\addplot [draw=symRowTwo, fill=symRowTwo!28, solid, line width=0.55pt,
  boxplot prepared={
    draw position=6,
    lower whisker=5.6341857,
    lower quartile=10.633075,
    median=14.883507,
    upper quartile=63.272773,
    upper whisker=192.60754
  }
] coordinates {};
\addplot [draw=symRowThree, fill=symRowThree!28, solid, line width=0.55pt,
  boxplot prepared={
    draw position=7,
    lower whisker=4.7091615,
    lower quartile=9.7280651,
    median=13.941227,
    upper quartile=58.456043,
    upper whisker=131.55951
  }
] coordinates {};
\addplot [draw=symRowFour, fill=symRowFour!28, solid, line width=0.55pt,
  boxplot prepared={
    draw position=8,
    lower whisker=4.4672225,
    lower quartile=8.4297475,
    median=13.686133,
    upper quartile=58.194528,
    upper whisker=107.82477
  }
] coordinates {};
\end{groupplot}
\end{tikzpicture}
    \caption{Runtime comparison of different symmetry-breaking options}
    \label{fig:sb-plots}
\end{figure}

In \Cref{fig:scatterplot}, we show another way to visualize the difference between symmetry-breaking methods using scatterplots for a ``paired'' comparison; each point corresponds to the same encoding choices run once with row-1 symmetry breaking and once with diagonal symmetry breaking. Thus, points below the diagonal dotted line indicate encodings for which the diagonal symmetry breaking is faster.
In these plots, we additionally used the \textsf{scranfilize} tool of Biere and Heule~\cite{scranfilize}, which randomly permutes the order of clauses, as well as other syntactic aspects of the formula such as the variable ordering. For each encoding setting, we ran 5 independent ``scranfilized'' formulas---each datapoint in the plot represents the average runtime of these 5 formulas. Under this view, it can be appreciated how the choice of within-cell encoding induces clear clusters in the data, whereas the choice of encoding for the between-cells constraints does not show such a clear tendency.  While for $n=8$ there is little benefit to using exactly-one encodings, their impact increases with $n$, and for $n = 10$ they already account for one order of magnitude in runtime differences.

\begin{figure}[ht]
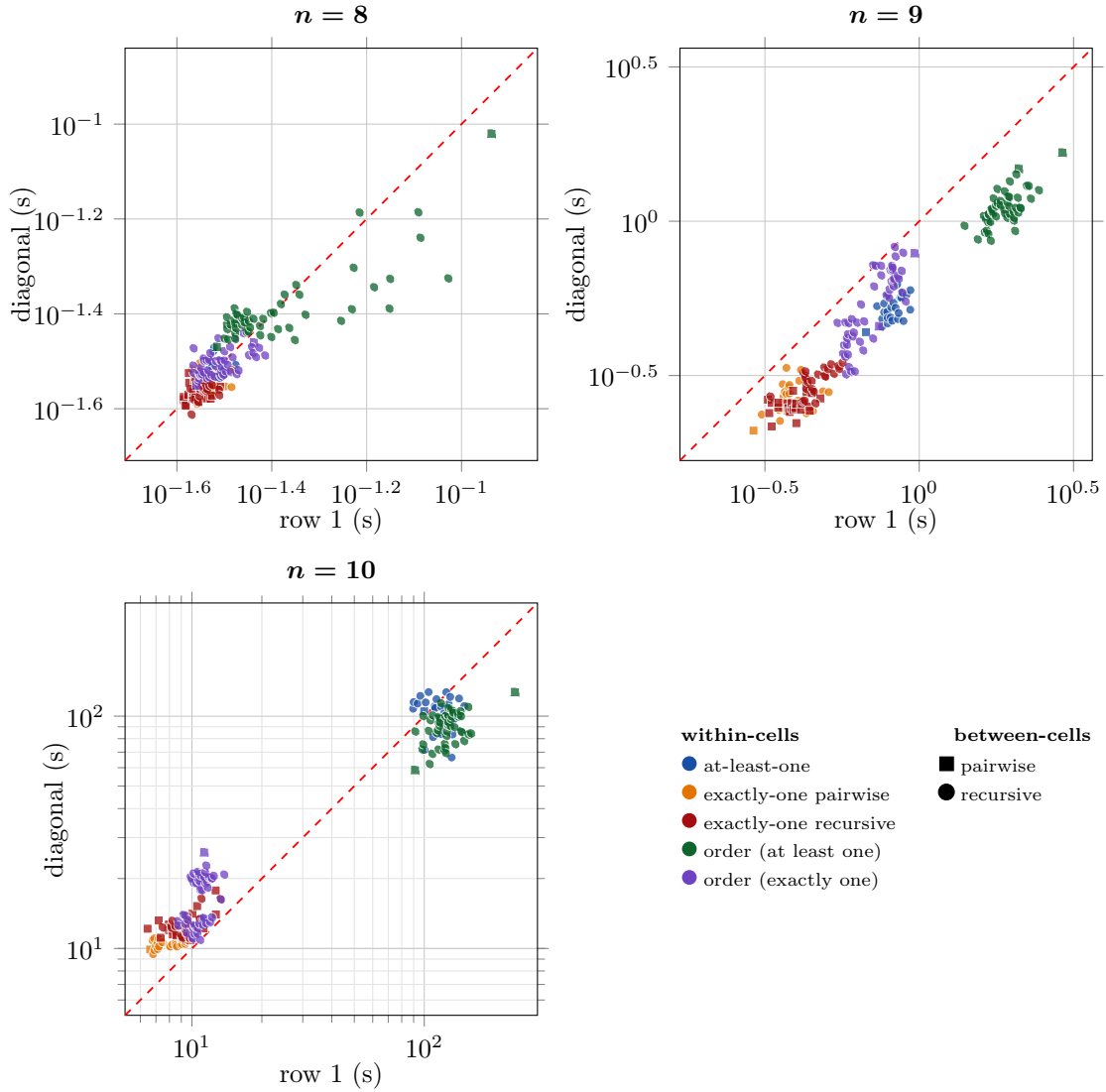

    \centering
    \begin{tikzpicture}
    \begin{groupplot}[
        group style={
            group size=2 by 2,
            horizontal sep=1.9cm,
            vertical sep=1.9cm,
        },
        width=5.5cm,
        height=5.5cm,
        scale only axis,
        title style={font=\bfseries\boldmath},
        xlabel={row 1 (s)},
        ylabel={diagonal (s)},
        xlabel style={yshift=5pt},
        ylabel style={yshift=-8pt},
        grid=both,
        major grid style={draw=gray!45, line width=0.35pt},
        minor grid style={draw=gray!22, line width=0.25pt},
        tick align=outside,
        legend style={, fill=none, mark options={scale=5.0}},
        legend cell align={left},
        clip=true,
        xmode=log,
        ymode=log,
    ]
    \input{figs/scatterplot_sb_8}
    \input{figs/scatterplot_sb_9}
    \input{figs/scatterplot_sb_10}
    \nextgroupplot[group/empty plot]
    \end{groupplot}
    \node at (group c2r2.center) {\scriptsize
      \begin{tabular}{@{}r@{\;}l@{\qquad}r@{\;}l@{}}
        \multicolumn{2}{@{}l}{\textbf{within-cells}} & \multicolumn{2}{l@{}}{\textbf{between-cells}} \\[3pt]
        \tikz\draw[fill=qcColor0, draw=white] (0,0) circle (3pt); & at-least-one &
        \tikz\node[draw=white, line width=0.25pt, fill=black, minimum size=6pt, inner sep=0pt] {}; & pairwise \\[2pt]
        \tikz\draw[fill=qcColor1, draw=white] (0,0) circle (3pt); & exactly-one pairwise &
        \tikz\draw[fill=black] (0,0) circle (3pt); & recursive \\[2pt]
        \tikz\draw[fill=qcColor2, draw=white] (0,0) circle (3pt); & exactly-one recursive & & \\[2pt]
        \tikz\draw[fill=qcColor4, draw=white] (0,0) circle (3pt); & order (at least one) & & \\[2pt]
        \tikz\draw[fill=qcColor3, draw=white] (0,0) circle (3pt); & order (exactly one) & & \\
      \end{tabular}
    };
    \end{tikzpicture}
    \caption{Scatterplots comparing row 1 versus diagonal symmetry breaking}
    \label{fig:scatterplot}
\end{figure}

\subsection{Clique Hints}

Since the usefulness of clique hints is one of our main experimental findings, we conducted additional experiments to analyze their impact further.  Namely, we considered specific sets of encoding choices, and instead of adding clique hints for each clique of maximum size and color, we only added them with a certain probability $p$ (independently for each clique hint). Thus, $p = 0$ corresponds to no clique hints, and $p = 1$ corresponds to adding all clique hints. As can be observed in~\Cref{fig:clique-hints}, which corresponds to $n = 9$ with the pairwise encoding for between-cells constraints and the exactly one pairwise encoding for within-cells, the runtime decreases exponentially as $p$ increases.

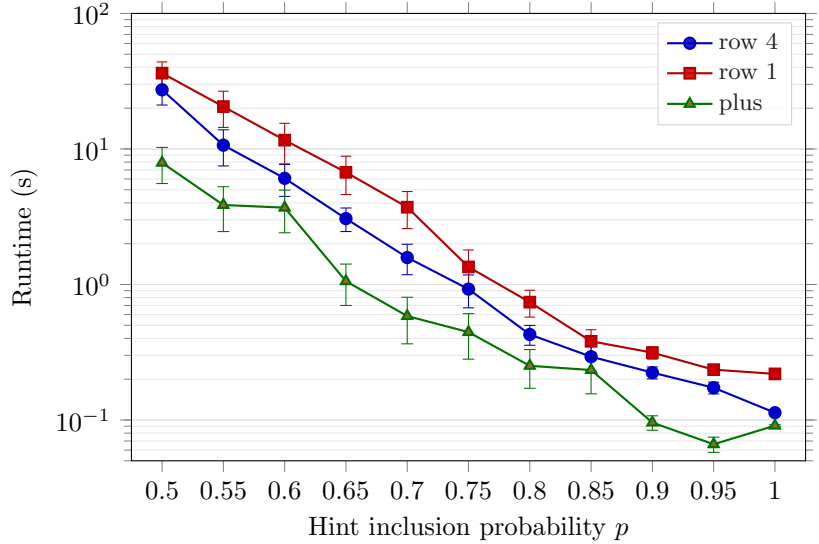
\begin{figure}
    \centering
    \begin{tikzpicture}
\pgfplotstableread{
p mean ci hints timeouts
0.50 27.323894002 6.230587320 89.000 0
0.55 10.668669446 3.175324043 101.100 0
0.60 6.073544681 1.607565766 107.150 0
0.65 3.062884973 0.604269579 116.300 0
0.70 1.580649944 0.400208149 128.300 0
0.75 0.924241800 0.251644153 134.150 0
0.80 0.427598398 0.071466677 143.100 0
0.85 0.293215354 0.063636462 153.100 0
0.90 0.223728160 0.022724357 161.600 0
0.95 0.173030965 0.017355942 171.600 0
1.00 0.113229139 0.000775459 180.000 0
}\datarowfour

\pgfplotstableread{
p mean ci hints timeouts
0.50 36.255142623 7.617560768 89.000 0
0.55 20.544462273 6.127555402 101.100 0
0.60 11.618079994 3.851575249 107.150 0
0.65 6.728129744 2.117613017 116.300 0
0.70 3.717949006 1.133202824 128.300 0
0.75 1.346969156 0.450696396 134.150 0
0.80 0.740842094 0.165918846 143.100 0
0.85 0.380955210 0.082459304 153.100 0
0.90 0.314152525 0.032791113 161.600 0
0.95 0.235225704 0.012532250 171.600 0
1.00 0.218890613 0.006110313 180.000 0
}\datarowone

\pgfplotstableread{
p mean ci hints timeouts
0.50 7.911753771 2.353588381 89.000 0
0.55 3.861150627 1.403803647 101.100 0
0.60 3.688035352 1.280273486 107.150 0
0.65 1.057415606 0.357132661 116.300 0
0.70 0.584891552 0.220056872 128.300 0
0.75 0.445143246 0.163973734 134.150 0
0.80 0.251147781 0.079785787 143.100 0
0.85 0.233779504 0.077725564 153.100 0
0.90 0.095688964 0.011670066 161.600 0
0.95 0.066236335 0.008523792 171.600 0
1.00 0.090843835 0.001445663 180.000 0
}\dataplus

\begin{axis}[
    width=10.5cm,
    height=7.5cm,
    ymode=log,
    ymin=0.05,
    ymax=100,
    xmin=0.475,
    xmax=1.025,
    xtick={0.50,0.55,0.60,0.65,0.70,0.75,0.80,0.85,0.90,0.95,1.00},
    xticklabels={0.5,0.55,0.6,0.65,0.7,0.75,0.8,0.85,0.9,0.95,1},
    xlabel={Hint inclusion probability $p$},
    ylabel={Runtime (s)},
    title={},
    ymajorgrids=true,
    yminorgrids=true,
    grid style={draw=gray!25},
    minor grid style={draw=gray!18},
    tick align=outside,
    legend style={
        at={(0.98,0.98)},
        anchor=north east,
        fill=white,
        fill opacity=0.88,
        draw=gray!35,
        font=\small,
    },
    legend cell align=left,
]

\addplot+[
    thick, mark=*, mark size=2.1pt, blue!75!black,
    error bars/.cd,
        y dir=both,
        y explicit,
]
table[x=p,y=mean,y error=ci] {\datarowfour};
\addlegendentry{row 4}

\addplot+[
    thick,  mark=square*, mark size=2.0pt, red!70!black,
    error bars/.cd,
        y dir=both,
        y explicit,
]
table[x=p,y=mean,y error=ci] {\datarowone};
\addlegendentry{row 1}

\addplot+[
    thick,  mark=triangle*, mark size=2.2pt, green!45!black,
    error bars/.cd,
        y dir=both,
        y explicit,
]
table[x=p,y=mean,y error=ci] {\dataplus};
\addlegendentry{plus}

\end{axis}
\end{tikzpicture}
    \caption{Effect of the inclusion of clique hints for $n = 9$, using $20$ independent trials per datapoint. Bars represent $95\%$ confidence intervals.}\label{fig:clique-hints}
\end{figure}

\subsection{On Brittleness}

While our statistical analysis of~\Cref{subsec:regression} measures the impact of different encoding choices \emph{on average}, we found some interesting cases in specific instances where encoding choices that generally do not have much impact showed large effects. For example, while for $n=9$, according to~\Cref{tab:regression-9}, the impact of the auxiliary placement is not even statistically significant (i.e., the confidence interval contains both positive and negative values), we found that in the absence of clique hints there are instances for which it makes a very large difference: using recursive cut-3 for the within-cell constraints, pairwise for the between-cell constraints, and row-1 symmetry breaking, the linear variant took 216 seconds whereas the pooled variant took 366 seconds. When using plus symmetry breaking, the linear variant took 34 seconds, whereas the pooled variant took 95. Furthermore, we observed that in these instances the \textsf{Kissat} flags \texttt{--forcephase=1 --phase=0} had a significant impact, which was not the case on average when applied to the entire dataset. This data is presented in~\Cref{tab:brittleness-aux-placement}.

\begin{table}[t]
\centering
\caption{Selected $n=9$ runtimes without clique hints, highlighting the brittleness of auxiliary placement.}
\label{tab:brittleness-aux-placement}
\begin{tabular}{llrr}
\toprule
 & & \multicolumn{2}{c}{Runtime (s)} \\
\cmidrule(l){3-4}
Symmetry breaking & Auxiliary placement & Default & Phase flags\\
\midrule
Row 1 & Linear & 216.72 &  43.31 \\
Row 1 & Pooled & 366.94 & 248.74 \\
Plus  & Linear & 34.75 & 31.07 \\
Plus  & Pooled & 95.35  & 98.23 \\
\bottomrule
\end{tabular}
\end{table}

Regarding a different form of brittleness, we also observed some interesting effects of clause ordering and other syntactic parameters (e.g., variable numbering). While the ordering of clauses in a CNF file makes no semantic difference, it has been documented already that the ordering of clauses can make an important difference in runtime~\cite{scranfilize, nikolicStatisticalMethodologyComparison2010a}. We observed a curious case of this phenomenon: due to the vertical symmetry of $n \times n$ grids, the formulas resulting from enforcing symmetry breaking on row $i$ are isomorphic to those enforcing symmetry breaking on row $n+1-i$ (which is why our bulk of experiments only considered rows $1$ to $\lceil n/2\rceil$), and yet, additional experiments including all possible rows reveal significant runtime differences between complementary pairs. As can be observed in~\Cref{fig:sym-reversal} (corresponding to $n=10$, with a pairwise encoding between cells, and exactly one pairwise within cells) the effect of random shuffling is not uniform: it makes, e.g., row 1 significantly better but row 9 significantly worse.  In general, without scranfilizing, the results are even further from symmetric around the middle, and in fact, the bottom 5 rows turn out to be roughly equivalent.

\begin{figure}
    \centering
    \begin{tikzpicture}
\pgfplotstableread{
row scran scranCi normal
1 5.8704 0.5469 7.3607
2 5.5246 0.3446 5.0488
3 6.0787 0.6031 4.0348
4 5.0843 0.6696 3.2284
5 4.5757 0.3944 4.4785
6 4.5585 0.2205 4.3023
7 5.6800 0.7653 4.3259
8 5.6191 0.2788 4.5218
9 7.0931 0.6240 4.7686
10 5.5693 0.2136 4.7675
}\symdata

\begin{axis}[
    width=14cm,
    height=5.5cm,
    ybar=0pt,
    bar width=8.0pt,
    ymin=0,
    ymax=8.5000,
    xmin=0.5,
    xmax=10.65,
    xtick={1,2,3,4,5,6,7,8,9, 10},
    xlabel={Symmetry-breaking row index},
    ylabel={Runtime (s)},
    title={},
    ymajorgrids=true,
    grid style={draw=qcGrid, line width=0.35pt},
    tick align=outside,
    axis line style={draw=qcInk},
    tick style={draw=qcInk},
    label style={font=\small},
    tick label style={font=\small},
    title style={font=\large\bfseries, yshift=1mm},
    legend style={
        at={(0.4,0.98)},
        anchor=north,
        legend columns=2,
        draw=none,
        fill=none,
        font=\small,
    },
    legend cell align=left,
    clip=false,
]

% \fill[qcGreen!8] (axis cs:3.5,0) rectangle (axis cs:5.5,8.5000);
% \node[anchor=south, font=\scriptsize\bfseries, text=qcGreen!60!black]
%     at (axis cs:4.5,8.0750) {central rows};

\addplot+[
    fill=blue!50,
    draw=blue!55!black,
    line width=0.45pt,
    error bars/.cd,
        y dir=both,
        y explicit,
        error bar style={draw=qcInk, line width=0.55pt},
        error mark options={draw=qcInk, line width=0.55pt, rotate=90, mark size=2.2pt},
]
table[x=row,y=scran,y error=scranCi] {\symdata};
\addlegendentry{scranfilized (20)\;}

\addplot+[
    fill=red!60,
    draw=red!55!black,
    line width=0.45pt,
]
table[x=row,y=normal] {\symdata};
\addlegendentry{no scranfilize}

% \draw[qcGreen!50!black, line width=0.45pt, densely dashed]
%     (axis cs:4.5,0) -- (axis cs:4.5,8.5000);

\end{axis}
\end{tikzpicture}
    \caption{Effect of vertical symmetry on symmetry breaking by rows, for $n = 10$. When using~\textsf{scranfilize}, we considered 20 independent trials. Bars represent $95\%$ confidence intervals.}
    \label{fig:sym-reversal}
\end{figure}
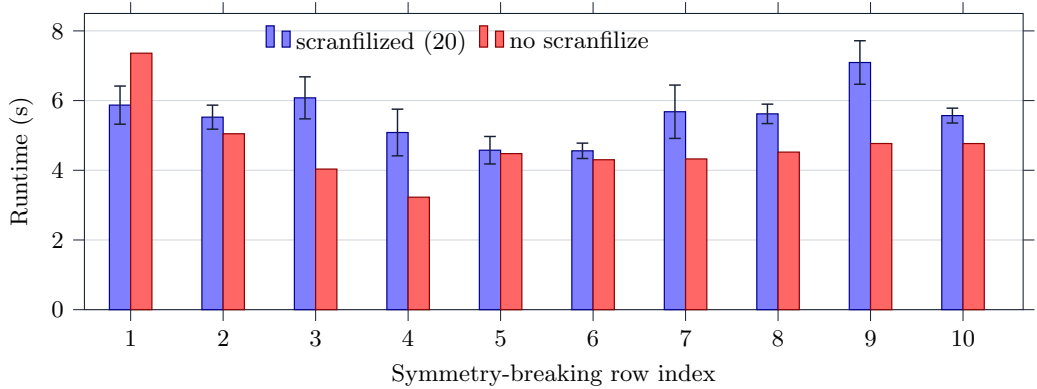

In general terms, we see these examples as a cautionary tale with regard to encodings, since they showcase how choices deemed minor can end up having large impacts on the runtime.

\section{Conclusion and Future Directions}

Besides yielding practical advice for how to encode graph coloring problems into SAT, we believe that our empirical methodology of studying SAT encodings ``at scale'' has the potential to improve our understanding of encodings more generally. Thus, we believe doing similar experiments for problems across several different domains would be fruitful. There are also many more parameters that could be varied in future analyses. For example, we kept the SAT solver settings fixed, but these can be important for some problems, and there could be interaction effects between encoding options and solver settings. Finally, we lack a framework for predicting and explaining how differences in encodings affect differences in runtimes. For instance, we would like to be able to explain (or even predict \emph{a priori}) why some symmetry-breaking methods are better than others and why their relative ordering changes between $n=9$ and $n=10$.

\subsection*{Acknowledgments}

This research is supported by the DARPA expMath program through the DARPA CMO contract number HR0011262E028 and by the
National Science Foundation (NSF) under grant DMS2434625.
Przybocki was additionally supported by the NSF Graduate Research Fellowship Program under Grant Nos. DGE-2140739 and DGE-2631988.

\bibliographystyle{plain}
\bibliography{references}

\end{document}